\documentclass[preprint,showpacs,preprintnumbers,amsmath,amssymb]{revtex4}
\usepackage{amsmath}

\usepackage{graphicx}
\usepackage{dcolumn}
\usepackage{bm}
\usepackage{amsmath}

\newtheorem{theorem}{Theorem}[section]

\newtheorem{proposition}[theorem]{Proposition}

\newenvironment{proof}[1][Proof]{\begin{trivlist}
\item[\hskip \labelsep {\bfseries #1}]}{\end{trivlist}}

\newcommand{\qed}{\nobreak \ifvmode \relax \else
      \ifdim\lastskip<1.5em \hskip-\lastskip
      \hskip1.5em plus0em minus0.5em \fi \nobreak
      \vrule height0.75em width0.5em depth0.25em\fi}

\begin{document}

\preprint{}

\title{Separability Probability Formulas and Their Proofs for Generalized Two-Qubit 
X-Matrices Endowed with Hilbert-Schmidt and Induced Measures} 
\author{Charles F. Dunkl}
 \email{cfd5z@virginia.edu}
\affiliation{Department of Mathematics, University of Virginia,
Charlottesville, VA 22904-4137}
\author{Paul B. Slater}
 \email{slater@kitp.ucsb.edu}
\affiliation{%
University of California, Santa Barbara, CA 93106-4030\\
}%
\date{\today}

\begin{abstract}
Two-qubit X-matrices have been the subject of considerable recent attention, as they lend themselves more readily to analytical investigations than two-qubit density matrices of arbitrary nature. 
Here, we maximally exploit this relative ease of analysis to formally derive an exhaustive collection of results pertaining to the separability probabilities of generalized two-qubit X-matrices endowed with Hilbert-Schmidt and, more broadly, induced measures. Further, the analytical results obtained  exhibit interesting parallels to corresponding earlier (but, contrastingly, not yet fully rigorous) results for general 2-qubit states--deduced on the basis of determinantal moment formulas. Geometric interpretations can be given to arbitrary positive values of the
random-matrix Dyson-index-like parameter $\alpha$ employed.
\end{abstract}

\pacs{Valid PACS 03.67.Mn, 02.30.Zz, 02.50.Cw, 02.40.Ft}
\keywords{X-states, $2 \times 2$ quantum systems, separability probabilities,  entanglement  probability distribution moments, probability distribution reconstruction, Peres-Horodecki conditions,  partial transpose, determinant of partial transpose, two qubits, two rebits, two quaterbits, Hilbert-Schmidt measure, induced measure, continuous Dyson index, 
moments, determinantal moments, inverse problems, random matrix theory, generalized two-qubit systems}

\maketitle
\section{Introduction}
In previous work \cite{MomentBased,concise}, the authors investigated the separability probabilities 
of $4\times 4$ density matrices ($\rho$) with density approximation techniques \cite{Provost} based on the moments  of the determinant  ($|\rho^{PT}|$) of the partial transpose--the nonnegativity of which is necessary and sufficient for separability \cite{peres,horodecki}. The moment formulas employed involve a Dyson-index-like parameter $\alpha$ which
specializes to real, complex, quaternionic densities for 
$\alpha=\frac{1}{2},1$ and 2, respectively. As yet, the validity of the moment formulas is only a
conjecture, but there is strong evidence that they are correct. By use of extraordinarily large ($5\times10^{11}$ trials) Monte Carlo simulations, Fei and
Joynt \cite{FJ} obtained separability probability values agreeing to 4 decimal places with
our results (cf. \cite{MiSt,dubna}), that is $\frac{29}{64}, \frac{8}{33}$ and $\frac{26}{323}$, respectively. More generally still, we have also been 
investigating the general separability probability question \cite{ZHSL} when
$k$-th powers of the determinant ($|\rho|$) times the Hilbert-Schmidt (HS) measure \cite{ingemarkarol,szHS}--that is, induced measure \cite{ingemarkarol,Induced,aubrun}--are used as the probability measure on the convex set of density matrices. Due to the extreme intractability of performing direct probability calculations by use of high-dimensional
integration, these studies continue to be based on moment formulas and
sophisticated numerical analysis (density approximation) 
techniques \cite{Provost} to determine probabilities.

Recently, interesting results have been obtained for the subset of so-called
X-matrices, which are a form of toy model for density matrices \cite{MiSt,Xstates1,Xstates2}. The present
study exploits their simpler structure to directly, and now fully rigorously, compute the separability
probabilities for the family of induced measures (including the particular [$k=0$] Hilbert-Schmidt case). Although the results obtained are numerically
quite different from those \cite{formulas} holding for the full, unrestricted matrices, they do exhibit a strong
qualitative similarity. This can be taken as evidence that our earlier work \cite{MomentBased,concise,formulas} has been well directed. 
We note that in the X-matrix case, it has been possible to give geometric
interpretations to arbitrary positive values of the Dyson-index-like parameter $\alpha$, not just integer and half-integer values, such as $\frac{1}{2},1$ and 2 as described above (cf. \cite{edelman}), and perhaps this idea can be extended to  unrestricted sets of density matrices.

We start with the basic definitions. A $4\times4$ X-density matrix $\xi$ is
positive semi-definite with $Tr\left(  \xi\right)  =1$ and has the form%
\[
\xi=\left[
\begin{array}
[c]{cccc}%
\xi_{11} & 0 & 0 & \xi_{14}\\
0 & \xi_{22} & \xi_{23} & 0\\
0 & \overline{\xi_{23}} & \xi_{33} & 0\\
\overline{\xi_{14}} & 0 & 0 & \xi_{44}%
\end{array}
\right]  ;
\]
the defining conditions are equivalent to $\xi_{jj}\geq0$ for $1\leq j\leq4$,
$\xi_{14},\xi_{23}\in\mathbb{C}$, $\sum_{j=1}^{4}\xi_{jj}=1$, $\left\vert
\xi_{14}\right\vert ^{2}\leq\xi_{11}\xi_{44}$ and $\left\vert \xi
_{23}\right\vert ^{2}\leq\xi_{22}\xi_{33}$. The partial transpose of $\xi$ is
denoted by $\xi^{PT}$ and
\[
\xi^{PT}=\left[
\begin{array}
[c]{cccc}%
\xi_{11} & 0 & 0 & \xi_{23}\\
0 & \xi_{22} & \xi_{14} & 0\\
0 & \overline{\xi_{14}} & \xi_{33} & 0\\
\overline{\xi_{23}} & 0 & 0 & \xi_{44}%
\end{array}
\right]
\]
Thus $\det\xi^{PT}=\left(  \xi_{11}\xi_{44}-\left\vert \xi_{23}\right\vert
^{2}\right)  \left(  \xi_{22}\xi_{33}-\left\vert \xi_{14}\right\vert
^{2}\right)  $ and the density matrix $\xi$  is separable if and only if $\det\xi^{PT}\geq0$.

We induce a measure on the set $\mathcal{X}$ of X-density matrices from the
measure%
\[
\prod_{j=1}^{4}\mathrm{d}\xi_{jj}\left(  r_{5}r_{6}\right)  ^{2\alpha
-1}\mathrm{d}r_{5}\mathrm{d}r_{6}\mathrm{d}\theta_{5}\mathrm{d}\theta_{6},
\]
on the cone of all positive semi-definite X-matrices, where $\xi_{14}%
=r_{5}e^{\mathrm{i}\theta_{5}},\xi_{23}=r_{6}e^{\mathrm{i}\theta_{6}}$
($r_{j}\geq0,~-\pi<\theta_{j}\leq\pi,~j=5,6$). The details of normalizing the
measure to be a probability measure appear later. We will show for any
$\alpha>0$ that%
\[
\Pr\left\{  \det\xi^{PT}\geq0\right\}  =p\left(  \alpha\right)  :=\frac
{1}{\sqrt{2\pi}}\frac{\Gamma\left(  \alpha+\frac{1}{2}\right)  ^{2}%
\Gamma\left(  \alpha+\frac{3}{2}\right)  }{\Gamma\left(  \alpha+\frac{3}%
{4}\right)  \Gamma\left(  \alpha+1\right)  \Gamma\left(  \alpha+\frac{5}%
{4}\right)  },
\]
in particular, $p\left(  \frac{1}{2}\right)  =\dfrac{16}{3\pi^{2}},~p\left(
1\right)  =\dfrac{2}{5},~p\left(  2\right)  =\dfrac{2}{7}$. The value
$p\left(  1\right)  =\frac{2}{5}$ is also found  in
Milz and Strunz \cite[eq. (22)]{MiSt}, where they employed a quite different, interesting 
analytical framework, in which the principal variable of interest 
is the radial location in the Bloch ball of the reduced density matrix.

Further, we find $\Pr\left\{  \det\xi^{PT}\geq0\right\}  $ when the measure is
multiplied by $\left(  \det\xi\right)  ^{k}$ for some fixed $k\geq0$ and
$\alpha=1,2,3,\ldots$. For example when $\alpha=1$, $\Pr\left\{  \det\xi
^{PT}\geq0\right\}  =1-\frac{2\Gamma\left(  2k+4\right)  ^{2}}{\Gamma\left(
k+2\right)  \Gamma\left(  3k+6\right)  }$. For $\alpha=2,3,\ldots$we find that
$\Pr\left\{  \det\xi^{PT}\leq0\right\}  $ is a product of a ratio of gamma
functions with a polynomial in $k$ with positive coefficients of degree
$2\alpha-3$. This combination is similar to our results for the full
$4\times4$ density matrices \cite{formulas}.

Section 2 contains the construction of the coordinate system used to set up
the relevant definite integrals as well as the expression of the normalized
measure in this system, and explains the relation of the values $\alpha
=\frac{1}{2},1,2$ to real, complex, and quaternionic matrices.

In Section 3 the computation of $\Pr\left\{  \det\xi^{PT}\geq0\right\}  $ for
any $\alpha>0$ is carried out. The computation starts with a five-fold
iterated integral which is reduced to an integral of hypergeometric type and
finally to a classical hypergeometric summation formula.

The measure of $\left(  \det\xi\right)  ^{k}$ type is studied in Section 4.
The calculation of $\Pr\left\{  \det\xi^{PT}\leq0\right\}  $ is carried out
only for integer values of $\alpha$ for reasons of technical difficulty. As
will be seen, even this integer case is quite complicated. Necessary integral
formulas are derived in Section \ref{intform}.

\section{The coordinate system and the measures}

We construct a family of probability measures on $\mathcal{X}$, with a
parameter $\alpha$, which agree with the normalized measures induced by the
Hilbert-Schmidt measures on $M_{4}\left(  \mathbb{R}\right)  ,M_{4}\left(
\mathbb{C}\right)  ,M_{4}\left(  \mathbb{H}\right)  $ for $\alpha=\frac{1}%
{2},1,2$ respectively ($\mathbb{H}$ denotes the quaternions). As in our
previous studies \cite{MomentBased,formulas,wholehalf}, we use the Cholesky decomposition to define the measures. For
$x\in\mathbb{R}_{\geq0}^{4}\times\mathbb{C}^{2}$ let%
\begin{align*}
C  &  =\left[
\begin{array}
[c]{cccc}%
x_{1} & 0 & 0 & x_{5}\\
0 & x_{2} & x_{6} & 0\\
0 & 0 & x_{3} & 0\\
0 & 0 & 0 & x_{4}%
\end{array}
\right]  ,\\
\xi &  =C^{\ast}C=\left[
\begin{array}
[c]{cccc}%
x_{1}^{2} & 0 & 0 & x_{1}x_{5}\\
0 & x_{2}^{2} & x_{2}x_{6} & 0\\
0 & x_{2}\overline{x_{6}} & x_{3}^{2}+\left\vert x_{6}\right\vert ^{2} & 0\\
x_{1}\overline{x_{5}} & 0 & 0 & x_{4}^{2}+\left\vert x_{5}\right\vert ^{2}%
\end{array}
\right]  ,
\end{align*}
eventually we will impose the restriction $\sum_{j=1}^{4}x_{j}^{2}+\left\vert
x_{5}\right\vert ^{2}+\left\vert x_{6}\right\vert ^{2}=1$ so that $trC=1$, but
first we work with arbitrary positive-definite matrices.

To motivate the definition of the measures, we describe a simple model for the
$2\times2$ block. Consider the map defined on the subset $H:=\mathbb{R}%
_{+}^{2}\times\mathbb{R}^{m}$ of $\mathbb{R}^{2+m}$ (for some fixed
$m=1,2,\ldots$)%
\[
\phi:\left(  x_{1},x_{2},y_{1},\cdots,y_{m}\right)  \mapsto\left(  x_{1}%
^{2},x_{2}^{2}+\left\vert y\right\vert ^{2},x_{1}y_{1},\ldots,x_{1}%
y_{m}\right)  ;
\]
The Euclidean measure on $H$ can be expressed as $t_{1}^{-1/2}t_{2}%
^{-1/2}t_{3}^{m/2-1}\mathrm{d}t_{1}\mathrm{d}t_{2}\mathrm{d}t_{3}%
\mathrm{d}\omega\left(  y^{\prime}\right)  $ where $t_{1}=x_{1}^{2}%
,t_{2}=x_{2}^{2}$ and $t_{3}=\left\vert y\right\vert ^{2},y=\left\vert
y\right\vert y^{\prime}$ and $\mathrm{d}\omega\left(  y^{\prime}\right)  $ is
the surface measure on the unit sphere $\left\{  y^{\prime}\in\mathbb{R}%
^{m}:\left\vert y^{\prime}\right\vert ^{2}=1\right\}  $. The point $\left(
t_{1},t_{2},t_{3}\right)  \in\mathbb{R}_{+}^{3}$. The Jacobian of $\phi$
equals $4x_{1}^{1+m}x_{2}$. The image of the measure on $H$ under $\phi$ is
(constants are discarded)%
\[
t_{1}^{m/2}t_{3}^{m/2-1}\mathrm{d}t_{1}\mathrm{d}t_{2}\mathrm{d}%
t_{3}\mathrm{d}\omega\left(  y^{\prime}\right)  .
\]
Now adjoin another copy of $H$ and the map (a direct sum) and relabel to
arrive at%
\[
\phi:\left(  x_{1},x_{2},x_{3},x_{4},y^{\left(  5\right)  },y^{\left(
6\right)  }\right)  \in\mathbb{R}_{+}^{4}\times\mathbb{R}^{2m}\mapsto\left(
x_{1}^{2},x_{2}^{2},x_{3}^{2}+\left\vert y^{\left(  6\right)  }\right\vert
^{2},x_{4}^{2}+\left\vert y^{\left(  5\right)  }\right\vert ^{2}%
,x_{1}y^{\left(  5\right)  },x_{2}y^{\left(  6\right)  }\right)  ,
\]
and the measure
\[
t_{1}^{m/2}t_{2}^{m/2}t_{5}^{m/2-1}t_{6}^{m/2-1}\mathrm{d}t_{1}\mathrm{d}%
t_{2}\mathrm{d}t_{3}\mathrm{d}t_{4}\mathrm{d}t_{5}\mathrm{d}t_{6}%
d\omega\left(  y^{\left(  5\right)  \prime}\right)  d\omega\left(  y^{\left(
6\right)  \prime}\right)  .
\]
In the cases $m=1,2,4$ this construction can be interpreted in terms of the
Cholesky decomposition of a $4\times4$ positive-definite X-matrix over
$\mathbb{R},\mathbb{C},\mathbb{H}$ respectively. In this situation $\det
\xi^{PT}=\left(  t_{1}t_{4}+t_{1}t_{5}-t_{2}t_{6}\right)  \left(  t_{2}%
t_{3}+t_{2}t_{6}-t_{1}t_{5}\right)  $ and the $d\omega$ factor does not enter
into the calculation, and so is replaced by $1$. The same result is obtained
if $\xi_{14},\xi_{23}$ are replaced by $\left(  t_{1}t_{5}\right)
^{1/2}\mathrm{e}^{\mathrm{i}\theta_{5}},\left(  t_{2}t_{6}\right)
^{1/2}\mathrm{e}^{\mathrm{i}\theta_{6}}$ respectively (with $-\pi<\theta
_{5},\theta_{6}\leq\pi$) and $d\omega\left(  y^{\left(  5\right)  \prime
}\right)  d\omega\left(  y^{\left(  6\right)  \prime}\right)  $ is replaced by
$\left(  \frac{1}{2\pi}\right)  ^{2}\mathrm{d}\theta_{5}\mathrm{d}\theta_{6}$.
To sum up this discussion, the generic $4\times4$ positive-definite complex
X-matrix is%
\[
\left[
\begin{array}
[c]{cccc}%
t_{1} & 0 & 0 & \left(  t_{1}t_{5}\right)  ^{1/2}\mathrm{e}^{\mathrm{i}%
\theta_{5}}\\
0 & t_{2} & \left(  t_{2}t_{6}\right)  ^{1/2}\mathrm{e}^{\mathrm{i}\theta_{6}}
& 0\\
0 & \left(  t_{2}t_{6}\right)  ^{1/2}\mathrm{e}^{-\mathrm{i}\theta_{6}} &
t_{3}+t_{6} & 0\\
\left(  t_{1}t_{5}\right)  ^{1/2}\mathrm{e}^{-\mathrm{i}\theta_{5}} & 0 & 0 &
t_{4}+t_{5}%
\end{array}
\right]
\]
and the parametrized measure is%
\[
t_{1}^{\alpha}t_{2}^{\alpha}t_{5}^{\alpha-1}t_{6}^{\alpha-1}\mathrm{d}%
t_{1}\mathrm{d}t_{2}\mathrm{d}t_{3}\mathrm{d}t_{4}\mathrm{d}t_{5}%
\mathrm{d}t_{6}\mathrm{d}\theta_{5}\mathrm{d}\theta_{6};
\]
which has geometric interpretations when $\alpha=\frac{m}{2},$ with $m=1,2,4$.
The last step is to induce this measure on $\left\{  \xi:Tr\xi=1\right\}  $,
that is, on the unit simplex $T_{5}$ in $\mathbb{R}^{5}$ ($t_{2}=1-t_{1}%
-\sum_{j=3}^{6}t_{j}$) and drop $\mathrm{d}t_{2}$ from the measure. Also since
we are concerned with only $\det\xi^{PT}$ and $\det\xi$ we also drop
$\mathrm{d}\theta_{5}\mathrm{d}\theta_{6}$.

We have arrived at the measure on $T_{5}$%
\[
\mathrm{d}\mu_{\alpha}=t_{1}^{\alpha}t_{2}^{\alpha}t_{5}^{\alpha-1}%
t_{6}^{\alpha-1}\mathrm{d}t_{1}\mathrm{d}t_{3}\mathrm{d}t_{4}\mathrm{d}%
t_{5}\mathrm{d}t_{6},
\]
with normalization constant (a Dirichlet integral)%
\begin{align*}
c_{\alpha} &  =\frac{\Gamma\left(  4\alpha+4\right)  }{\Gamma\left(
\alpha+1\right)  ^{2}\Gamma\left(  \alpha\right)  ^{2}},\\
c_{\alpha}\int_{T^{5}}d\mu_{\alpha} &  =1.
\end{align*}
With these coordinates
\begin{align*}
\det\xi &  =t_{1}t_{2}t_{3}t_{4},\\
\det\xi^{PT} &  =\left(  t_{1}t_{4}+t_{1}t_{5}-t_{2}t_{6}\right)  \left(
t_{2}t_{3}+t_{2}t_{6}-t_{1}t_{5}\right)  .
\end{align*}
We introduce the desired coordinate system in two steps. The first step is:%
\begin{align*}
\xi_{11} &  =t_{1}=\frac{1}{2}\left(  1-s_{1}+s_{2}\right)  ,\xi_{22}%
=t_{2}=\frac{1}{2}\left(  s_{1}+s_{3}\right)  ,\\
\xi_{33} &  =t_{3}+t_{6}=\frac{1}{2}\left(  s_{1}-s_{3}\right)  ,\xi
_{44}=t_{4}+t_{5}=\frac{1}{2}\left(  1-s_{1}-s_{2}\right)  ,\\
t_{5} &  =\frac{1}{2}s_{4}\left(  1-s_{1}-s_{2}\right)  ,t_{6}=\frac{1}%
{2}s_{5}\left(  s_{1}-s_{3}\right)  ,\\
t_{3} &  =\frac{1}{2}\left(  1-s_{5}\right)  \left(  s_{1}-s_{3}\right)
,t_{4}=\frac{1}{2}\left(  1-s_{4}\right)  \left(  1-s_{1}-s_{2}\right)  ,
\end{align*}
where $0\leq s_{1},s_{4},s_{5}\leq1$ and $\left\vert s_{2}\right\vert
\leq1-s_{1},\left\vert s_{3}\right\vert \leq s_{1}$. The Jacobian (omitting
$t_{2}$ from the list $\left(  t\right)  $ ) is%
\[
\frac{\partial\left(  t\right)  }{\partial\left(  s\right)  }=\frac{1}%
{16}\left(  s_{1}-s_{3}\right)  \left(  1-s_{1}-s_{2}\right)  ,
\]
and the measure and $\det\xi^{PT}$ transform to%
\begin{align*}
\mathrm{d}\mu_{\alpha} &  =2^{-4\alpha-2}\left(  \left(  1-s_{1}\right)
^{2}-s_{2}^{2}\right)  ^{\alpha}\left(  s_{1}^{2}-s_{3}^{2}\right)  ^{\alpha
}s_{4}^{\alpha-1}s_{5}^{\alpha-1}\mathrm{d}s_{1}\mathrm{d}s_{2}\mathrm{d}%
s_{3}\mathrm{d}s_{4}\mathrm{d}s_{5},\\
\det\xi^{PT} &  =\frac{1}{16}\left(  \left(  \left(  1-s_{1}\right)
^{2}-s_{2}^{2}\right)  -s_{5}\left(  s_{1}^{2}-s_{3}^{2}\right)  \right)
\left(  \left(  s_{1}^{2}-s_{3}^{2}\right)  -s_{4}\left(  \left(
1-s_{1}\right)  ^{2}-s_{2}^{2}\right)  \right)  .
\end{align*}
Since $\det\xi^{PT}$ is even in $s_{2}$ and $s_{3}$ we can restrict to $0\leq
s_{2}\leq1-s_{1},~0\leq s_{3}\leq s_{1}$ (a \textquotedblleft
quarter\textquotedblright\ of $\mathcal{X}$, denoted by $\mathcal{X}_{0}$) and
multiply the measure by $4$. The second step is to change variables%
\begin{align}
s_{2} &  =\sqrt{\left(  1-s_{1}\right)  ^{2}-4\delta_{1}},s_{3}=\sqrt
{s_{1}^{2}-4\delta_{2}},\label{s2delta}\\
\frac{\partial\left(  s_{2},s_{3}\right)  }{\partial\left(  \delta_{1}%
,\delta_{2}\right)  } &  =\frac{4}{\sqrt{\left(  1-s_{1}\right)  ^{2}%
-4\delta_{1}}\sqrt{s_{1}^{2}-4\delta_{2}}},\nonumber\\
\det\xi^{PT} &  =\left(  \delta_{1}-s_{5}\delta_{2}\right)  \left(  \delta
_{2}-s_{4}\delta_{1}\right)  ,\nonumber\\
\det\xi &  =\delta_{1}\delta_{2}\left(  1-s_{4}\right)  \left(  1-s_{5}%
\right)  ,\nonumber\\
\det\xi^{PT}-\det\xi &  =\left(  \delta_{1}-\delta_{2}\right)  \left(
s_{5}\delta_{2}-s_{4}\delta_{1}\right)  .\label{detdet}%
\end{align}
The measure transforms to%
\begin{align*}
\mathrm{d}\nu_{\alpha} &  =\frac{\delta_{1}^{\alpha}\delta_{2}^{\alpha}%
s_{4}^{\alpha-1}s_{5}^{\alpha-1}}{\sqrt{\frac{1}{4}\left(  1-s_{1}\right)
^{2}-\delta_{1}}\sqrt{\frac{1}{4}s_{1}^{2}-\delta_{2}}}\mathrm{d}%
s_{1}\mathrm{d}\delta_{1}\mathrm{d}\delta_{2}\mathrm{d}s_{4}\mathrm{d}s_{5},\\
0 &  \leq s_{1},s_{4},s_{5}\leq1,~0\leq\delta_{1}\leq\left(  \frac{1-s_{1}}%
{2}\right)  ^{2},~0\leq\delta_{2}\leq\left(  \frac{s_{1}}{2}\right)  ^{2};
\end{align*}
There is a beta-integral which we will use later, and also to verify the
normalization (recall $B\left(  a,b\right)  :=\frac{\Gamma\left(  a\right)
\Gamma\left(  b\right)  }{\Gamma\left(  a+b\right)  }=\int_{0}^{1}%
v^{a-1}\left(  1-v\right)  ^{b-1}\mathrm{d}v$)%
\[
\int_{0}^{c}\frac{u^{\alpha}}{\sqrt{c-u}}du=c^{\alpha+1/2}B\left(
\alpha+1,\frac{1}{2}\right)
\]%
\begin{align}
\int_{\mathcal{X}_{0}}\mathrm{d}\nu_{\alpha} &  =\frac{2^{-2-4\alpha}}%
{\alpha^{2}}B\left(  \alpha+1,\frac{1}{2}\right)  ^{2}\int_{0}^{1}\left(
1-s_{1}\right)  ^{2\alpha+1}s_{1}^{2\alpha+1}\mathrm{d}s_{1}\label{normz}\\
&  =\frac{2^{-2-4\alpha}}{\alpha^{2}}\frac{\Gamma\left(  \alpha+1\right)
^{2}\pi\Gamma\left(  2\alpha+2\right)  ^{2}}{\Gamma\left(  \alpha+\frac{3}%
{2}\right)  ^{2}\Gamma\left(  4\alpha+4\right)  }=\frac{1}{c_{\alpha}%
},\nonumber
\end{align}
by use of the $\Gamma$-duplication formula $\Gamma\left(  2u\right)  =\frac
{1}{\sqrt{\pi}}2^{2u-1}\Gamma\left(  u\right)  \Gamma\left(  u+\frac{1}%
{2}\right)  $. We will also use the Pochhammer symbol, $\left(  t\right)
_{0}=1,\left(  t\right)  _{n+1}=\left(  t\right)  _{n}\left(  t+n\right)  $
for $t\in\mathbb{C}$.

We finish this section by describing the extreme values of $\det\xi$, $\det
\xi^{PT}$ and $\det\xi^{PT}-\det\xi$. The maximum value of both $\det\xi$ and
$\det\xi^{PT}$ is $\frac{1}{256}$ for $\xi=\frac{1}{4}I$ (identity matrix).
The minimum value of both $\det\xi^{PT}$ and $\det\xi^{PT}-\det\xi$ is
$-\frac{1}{16}$ for $s_{1}=s_{5}=0,s_{4}=1$,$\delta_{1}=\frac{1}{4},\delta
_{2}=0$ (and other matrices). To maximize $\det\xi^{PT}-\det\xi=\left(
\delta_{1}-\delta_{2}\right)  \left(  s_{5}\delta_{2}-s_{4}\delta_{1}\right)
$ set $s_{5}=1,s_{4}=0,\delta_{1}=\left(  \frac{1-s_{1}}{2}\right)
^{2},\delta_{2}=\left(  \frac{_{s_{1}}}{2}\right)  ^{2}$ to obtain $\frac
{1}{16}s_{1}^{2}\left(  1-2s_{1}\right)  $ with maximum value $\frac{1}{432}$
at $s_{1}=\frac{1}{3}$. These extreme X-state values are identical to those for the
full matrices.

\section{The computation of $\Pr\left\{  \det\xi^{PT}\geq0\right\}  $}

The desired probability is the $\nu_{\alpha}$-measure of the set\linebreak%
\ $\left\{  \left(  s_{1},\delta_{1},\delta_{2},s_{4},s_{5}\right)  :\left(
\delta_{1}-s_{5}\delta_{2}\right)  \left(  \delta_{2}-s_{4}\delta_{1}\right)
\geq0\right\}  $, in other words, the definite integral of $\mathrm{d}%
\nu_{\alpha}$ over this set. We start the iterated integral with the variables
$s_{4},s_{5}$. There are apparently two possibilities for $\left(  \delta
_{1}-s_{5}\delta_{2}\right)  \left(  \delta_{2}-s_{4}\delta_{1}\right)  \geq0$:

\begin{itemize}
\item $\delta_{1}-s_{5}\delta_{2}\leq0$ and $\delta_{2}-s_{4}\delta_{1}\leq0$;
this implies $\frac{\delta_{1}}{\delta_{2}}\leq s_{5}$ and $\frac{\delta_{2}%
}{\delta_{1}}\leq s_{4}$, but $\max\left(  \frac{\delta_{1}}{\delta_{2}}%
,\frac{\delta_{2}}{\delta_{1}}\right)  >1$ and one of the inequalities
contradicts $s_{4},s_{5}\leq1$, except for the trivial case $\delta_{1}%
=\delta_{2},s_{3}=1=s_{4}$, included in the following case (the products of
the two pairs of eigenvalues of $\xi^{PT}$ are $\left(  \delta_{1}-s_{5}%
\delta_{2}\right)  $ and $\left(  \delta_{2}-s_{4}\delta_{1}\right)  $, so
this demonstrates that $\xi^{PT}$ can have at most one negative eigenvalue);

\item $\delta_{1}-s_{5}\delta_{2}\geq0$ and $\delta_{2}-s_{4}\delta_{1}\geq0$;
equivalent to $\frac{\delta_{1}}{\delta_{2}}\geq s_{5}$ and $\frac{\delta_{2}%
}{\delta_{1}}\geq s_{4}$. Thus $\delta_{2}\leq\delta_{1}$ imposes the bounds
$0\leq s_{4}\leq\frac{\delta_{2}}{\delta_{1}}\leq1$ and $0\leq s_{5}\leq
1\leq\frac{\delta_{1}}{\delta_{2}}$. Similarly $\delta_{1}\leq\delta_{2}$
imposes the bounds $0\leq s_{4}\leq1$ and $0\leq s_{5}\leq\frac{\delta_{1}%
}{\delta_{2}}\leq1$. In both cases the integral of $s_{4}^{\alpha-1}%
s_{5}^{\alpha-1}\mathrm{d}s_{4}\mathrm{d}s_{5}$ over this region equals%
\[
\frac{1}{\alpha^{2}}\left(  \min\left(  \frac{\delta_{2}}{\delta_{1}}%
,\frac{\delta_{1}}{\delta_{2}}\right)  \right)  ^{\alpha}.
\]

\end{itemize}

As a side observation, the computation of $\Pr\left\{  \det\xi^{PT}\geq\det
\xi\right\}  =\Pr\left\{  \left(  \delta_{1}-\delta_{2}\right)  \left(
s_{5}\delta_{2}-s_{4}\delta_{1}\right)  \geq0\right\}  $ starts with
integrating $s_{4}^{\alpha-1}s_{5}^{\alpha-1}\mathrm{d}s_{4}\mathrm{d}s_{5}$
over the region $\delta_{1}\geq\delta_{2},0\leq s_{4}\leq\frac{\delta_{2}%
}{\delta_{1}}s_{5}$, or the region $\delta_{2}\geq\delta_{1},0\leq s_{5}%
\leq\frac{\delta_{1}}{\delta_{2}}s_{4}$; the result in both cases is $\frac
{1}{2\alpha^{2}}\left(  \min\left(  \frac{\delta_{2}}{\delta_{1}},\frac
{\delta_{1}}{\delta_{2}}\right)  \right)  ^{\alpha}$. Since the rest of the
probability calculation is the same for both, we see that $\Pr\left\{  \det
\xi^{PT}\geq\det\xi\right\}  =\frac{1}{2}\Pr\left\{  \det\xi^{PT}>0\right\}
$, analogously to the general $4\times4$ situation, as we discussed in earlier work 
\cite{wholehalf}.

By symmetry, it suffices to integrate over $\left\{  0\leq s_{1}\leq\frac{1}%
{2}\right\}  $ and double the result. Denote $a:=\left(  \frac{1-s_{1}}%
{2}\right)  ^{2}$ and $b:=\left(  \frac{s_{1}}{2}\right)  ^{2}$, thus $0\leq
b\leq a\leq\frac{1}{4}$. In the iterated triple integral first integrate with
respect to $\delta_{1}$. Let%
\begin{align*}
I_{\alpha} &  :=\frac{2}{\alpha^{2}}\int_{0}^{1/2}\mathrm{d}s_{1}\int_{0}%
^{b}\frac{\delta_{2}^{\alpha}}{\sqrt{b-\delta_{2}}}\mathrm{d}\delta
_{2}\left\{  \int_{0}^{\delta_{2}}\left(  \frac{\delta_{1}}{\delta_{2}%
}\right)  ^{\alpha}\frac{\delta_{1}^{\alpha}}{\sqrt{a-\delta_{1}}}%
\mathrm{d}\delta_{1}+\int_{\delta_{2}}^{a}\left(  \frac{\delta_{2}}{\delta
_{1}}\right)  ^{\alpha}\frac{\delta_{1}^{\alpha}}{\sqrt{a-\delta_{1}}%
}\mathrm{d}\delta_{1}\right\}  \\
&  =\frac{2}{\alpha^{2}}\int_{0}^{1/2}\mathrm{d}s_{1}\int_{0}^{b}%
\frac{\mathrm{d}\delta_{2}\ }{\sqrt{b-\delta_{2}}}\left\{  \int_{0}%
^{\delta_{2}}\frac{\delta_{1}^{2\alpha}}{\sqrt{a-\delta_{1}}}\mathrm{d}%
\delta_{1}+\delta_{2}^{2\alpha}\int_{\delta_{2}}^{a}\frac{\mathrm{d}\delta
_{1}}{\sqrt{a-\delta_{1}}}\right\}  .
\end{align*}
In the $\left\{  \cdot\right\}  $ expression, the second integral equals $2\sqrt{a-\delta_{2}%
}$; for the first one, interchange the order of integration (over $0\leq
\delta_{1}\leq\delta_{2}\leq b$) to obtain%
\[
\int_{0}^{b}\frac{\delta_{1}^{2\alpha}}{\sqrt{a-\delta_{1}}}\mathrm{d}%
\delta_{1}\int_{\delta_{1}}^{b}\frac{\mathrm{d}\delta_{2}\ }{\sqrt
{b-\delta_{2}}}=2\int_{0}^{b}\delta_{1}^{2\alpha}\sqrt{\frac{b-\delta_{1}%
}{a-\delta_{1}}}\mathrm{d}\delta_{1}.
\]
Thus,%
\[
I_{\alpha}=\frac{4}{\alpha^{2}}\int_{0}^{1/2}\mathrm{d}s_{1}\int_{0}%
^{b}v^{2\alpha}\left\{  \sqrt{\frac{b-v}{a-v}}+\sqrt{\frac{a-v}{b-v}}\right\}
\mathrm{d}v.
\]
This integral can be directly evaluated for any $\alpha$ with $2\alpha
\in\mathbb{Z}_{+}$ with elementary methods. However, this integral is one of a
parametrized family which all have closed forms for their values. The
integrals are denoted by
\[
I\left(  m,n\right)  :=\int_{0}^{\frac{1}{2}}ds\int_{0}^{b}v^{m}\left\{
\left(  a-v\right)  ^{n}\sqrt{\frac{a-v}{b-v}}+\left(  b-v\right)  ^{n}%
\sqrt{\frac{b-v}{a-v}}\right\}  \mathrm{d}v.
\]
The formulas for $\left(  m,n\right)  =\left(  2\alpha,0\right)  $ for
$\alpha>0$ and for $m,n=0,1,2,3,\ldots$ are derived in Section \ref{intform}.

From Proposition \ref{sum2a} it follows that
\[
I_{\alpha}=\frac{4}{\alpha^{2}}I\left(  2\alpha,0\right)  =\frac{\pi
}{2^{2+8\alpha}}\frac{\Gamma(2\alpha)\Gamma(2\alpha)}{\Gamma\left(
2\alpha+\frac{3}{2}\right)  \Gamma\left(  2\alpha+\frac{5}{2}\right)  }.%
\]

When $2\alpha\in\mathbb{Z}_{+}$ the reciprocal of $I_{\alpha}$ is an integer.
In fact,%
\[
I_{\alpha}^{-1}=2\left(  4\alpha+3\right)  \binom{2\alpha+2}{2}^{2}%
\binom{4\alpha+1}{2\alpha+2}^{2}.
\]

Finally, we obtain
\begin{align}
\Pr\left\{  \det\xi^{PT}\geq0\right\}   &  =c_{\alpha}I_{\alpha}=\frac{\pi
}{2^{2+8\alpha}}\frac{\Gamma\left(  4\alpha+4\right)  \Gamma(2\alpha
)\Gamma(2\alpha)}{\Gamma\left(  \alpha+1\right)  ^{2}\Gamma\left(
\alpha\right)  ^{2}\Gamma\left(  2\alpha+\frac{3}{2}\right)  \Gamma\left(
2\alpha+\frac{5}{2}\right)  }\label{p(alpha)}\\
&  =\frac{1}{\sqrt{2\pi}}\frac{\Gamma\left(  \alpha+\frac{1}{2}\right)
^{2}\Gamma\left(  \alpha+\frac{3}{2}\right)  }{\Gamma\left(  \alpha+\frac
{3}{4}\right)  \Gamma\left(  \alpha+1\right)  \Gamma\left(  \alpha+\frac{5}%
{4}\right)  }=:p\left(  \alpha\right)  \nonumber
\end{align}
by repeated use of the $\Gamma$-duplication formula. If $\alpha=m$ is an
integer, then $c_{\alpha}=\dfrac{\left(  m-1\right)  ^{2}\left(  4m+3\right)
!}{m!^{4}}$ and $p\left(  \alpha\right)  $ is rational. If $\alpha+\frac{1}%
{2}=m$, then $c_{\alpha}=\dfrac{4\left(  m-\frac{1}{2}\right)  ^{2}\left(
4m+2\right)  !}{\pi^{2}\left(  \frac{1}{2}\right)  _{m}}$ and $p\left(
\alpha\right)  \pi^{2}$ is rational. The asymptotic formula $\dfrac
{\Gamma\left(  u+a\right)  }{\Gamma\left(  u+b\right)  }\sim u^{a-b}$ as
$u\rightarrow\infty$ shows that $p\left(  \alpha\right)  \sim\dfrac{1}%
{\sqrt{2\pi\alpha}}$ as $\alpha\rightarrow\infty$. For example, $p\left(
10\right)  =0.12683$\ldots\ and $\frac{1}{\sqrt{20\pi}}=0.12616\ldots$.

\subsection{Minimally degenerate matrices}

In our previous work \cite{wholehalf}, we considered the separability probability of a minimally
degenerate density matrix, and found numerical and formulaic evidence that
this probability is exactly one half of the unrestricted one. We can prove
this relation in the present X-matrix case. From $\det\xi=\delta_{1}\delta
_{2}\left(  1-s_{4}\right)  \left(  1-s_{5}\right)  $, we see that a necessary
condition for minimal degeneracy is that one of $\delta_{1},\delta
_{2},\allowbreak\left(  1-s_{4}\right)  ,\allowbreak\left(  1-s_{5}\right)  $
vanishes. The possibility $\delta_{1}=0$ or $\delta_{2}=0$ is not minimal, for
suppose $\delta_{1}=0$ then (see (\ref{s2delta})) $s_{2}=1-s_{1}$ which
implies $t_{4}=0=t_{5}$, reducing the dimension by two. This leaves the cases
$s_{4}=1$ or $s_{5}=1$ (for example, $s_{5}=1$ is equivalent to $\xi_{22}%
\xi_{33}-\left\vert \xi_{23}\right\vert ^{2}=0$) which remove only one
dimension (degree of freedom). So the minimally degenerate subset of
$\mathcal{X}$ consists of two almost (i.e. the intersection has measure zero)
disjoint subsets, one with $s_{4}=1$ and one with $s_{5}=1$. It suffices to
compute the probability for one of these, say $s_{5}=1$. First we compute the
normalizing constant: the calculation is the same as in (\ref{normz}) except
for the factor $\frac{1}{\alpha}$ from $\int_{0}^{1}s_{5}^{\alpha-1}%
\mathrm{d}s_{5}$, thus the constant is $\frac{c_{\alpha}}{\alpha}$. We see
$\det\xi^{PT}=\left(  \delta_{1}-\delta_{2}\right)  \left(  \delta_{2}%
-s_{4}\delta_{1}\right)  \geq0$ if $\delta_{1}\geq\delta_{2}$ and $0\leq
s_{4}\leq\frac{\delta_{2}}{\delta_{1}}$, or $\delta_{1}<\delta_{2}$ and
$s_{4}\geq\frac{\delta_{2}}{\delta_{1}}>1$; as before the second case is
impossible. Thus the first step of integration (for $s_{4}$) yields $\frac
{1}{\alpha}\left(  \frac{\delta_{2}}{\delta_{1}}\right)  ^{\alpha}$ for
$\delta_{1}\geq\delta_{2}$ and zero otherwise. In contrast to the earlier
computation, there is no $s_{1}\leftrightarrow1-s_{1}$ symmetry so $0\leq
s_{1}\leq\frac{1}{2}$ and $\frac{1}{2}\leq s_{1}\leq1$ are handled separately.

If $0\leq s_{1}\leq\frac{1}{2}$ (so that $b=\left(  \frac{s_{1}}{2}\right)
^{2}\leq\left(  \frac{1-s_{1}}{2}\right)  ^{2}=a$), then the remaining triple
integral is
\begin{align*}
&  \frac{1}{\alpha}\int_{0}^{1/2}\mathrm{d}s_{1}\int_{0}^{b}\frac{\delta
_{2}^{\alpha}}{\sqrt{b-\delta_{2}}}\mathrm{d}\delta_{2}\int_{\delta_{2}}%
^{a}\left(  \frac{\delta_{2}}{\delta_{1}}\right)  ^{\alpha}\frac{\delta
_{1}^{\alpha}}{\sqrt{a-\delta_{1}}}\mathrm{d}\delta_{1}\\
&  =\frac{2}{\alpha}\int_{0}^{1/2}\mathrm{d}s_{1}\int_{0}^{b}\delta
_{2}^{2\alpha}\sqrt{\frac{a-\delta_{2}}{b-\delta_{2}}}\mathrm{d}\delta_{2}.
\end{align*}
If $\frac{1}{2}\leq s_{1}\leq1$ (so that $b\geq a\geq\delta_{1}\geq\delta_{2}%
$), then the remaining triple integral is
\begin{align*}
&  \frac{1}{\alpha}\int_{1/2}^{1}\mathrm{d}s_{1}\int_{0}^{a}\frac{\delta
_{2}^{\alpha}}{\sqrt{b-\delta_{2}}}\mathrm{d}\delta_{2}\int_{\delta_{2}}%
^{a}\left(  \frac{\delta_{2}}{\delta_{1}}\right)  ^{\alpha}\frac{\delta
_{1}^{\alpha}}{\sqrt{a-\delta_{1}}}\mathrm{d}\delta_{1}\\
&  =\frac{2}{\alpha}\int_{1/2}^{1}\mathrm{d}s_{1}\int_{0}^{a}\delta
_{2}^{2\alpha}\sqrt{\frac{a-\delta_{2}}{b-\delta_{2}}}\mathrm{d}\delta_{2}.
\end{align*}
In the latter expression change variables $s=1-s_{1}$ which interchanges $a$
and $b$ with the result%
\[
\frac{2}{\alpha}\int_{0}^{1/2}\mathrm{d}s\int_{0}^{b}\delta_{2}^{2\alpha}%
\sqrt{\frac{b-\delta_{2}}{a-\delta_{2}}}\mathrm{d}\delta_{2}%
\]
and adding the two parts leads to%
\[
\frac{2}{\alpha}\int_{0}^{1/2}\mathrm{d}s\int_{0}^{b}\delta_{2}^{2\alpha
}\left\{  \sqrt{\frac{b-\delta_{2}}{a-\delta_{2}}}+\sqrt{\frac{a-\delta_{2}%
}{b-\delta_{2}}}\right\}  \mathrm{d}\delta_{2}=\frac{2}{\alpha}I\left(
2\alpha,0\right)  .
\]
Thus%
\[
\Pr\left\{  \det\xi^{PT}\geq0\right\}  =\frac{c_{\alpha}}{\alpha}\frac
{2}{\alpha}I\left(  2\alpha,0\right)  =\frac{1}{2}p\left(  \alpha\right)  ,
\]
see equation (\ref{p(alpha)}).

\section{The measure $\left(  \det\xi\right)  ^{k}\mathrm{d}\nu_{\alpha}$}

Here we compute $\Pr\left\{  \det\xi^{PT}<0\right\}  $ when $\mathcal{X}$ is
furnished with the normalization of the measure $\left(  \det\xi\right)
^{k}\mathrm{d}\nu_{\alpha}$, for $\alpha,k\in\mathbb{Z}_{+}$. It appears
possible to carry out the calculations for fixed integers $k$ and arbitrary
$\alpha>0$, but with the goal of allowing $k$ as a free parameter, technical
factors impel us to restrict to integer $\alpha$. Recall $\det\xi=\delta
_{1}\delta_{2}\left(  1-s_{4}\right)  \left(  1-s_{5}\right)  $. The
normalization constant is%
\[
c_{\alpha,k}=\frac{\Gamma\left(  4\alpha+4k+4\right)  }{\Gamma\left(
k+\alpha+1\right)  ^{2}\Gamma\left(  k+1\right)  ^{2}\Gamma\left(
\alpha\right)  ^{2}},
\]
and the measure is%
\[
\frac{\delta_{1}^{\alpha+k}\delta_{2}^{\alpha+k}s_{4}^{\alpha-1}s_{5}%
^{\alpha-1}\left(  1-s_{4}\right)  ^{k}\left(  1-s_{5}\right)  ^{k}}%
{\sqrt{\frac{1}{4}\left(  1-s_{1}\right)  ^{2}-\delta_{1}}\sqrt{\frac{1}%
{4}s_{1}^{2}-\delta_{2}}}\mathrm{d}s_{1}\mathrm{d}\delta_{1}\mathrm{d}%
\delta_{2}\mathrm{d}s_{4}\mathrm{d}s_{5}.
\]
Some experimenting suggests that it is more tractable to compute $\Pr\left\{
\det\xi^{PT}<0\right\}  $. Then, the first step is to compute (with $\delta
_{0}=\min\left(  \frac{\delta_{1}}{\delta_{2}},\frac{\delta_{2}}{\delta_{1}%
}\right)  $)%
\[
\int_{0}^{1}u^{\alpha-1}\left(  1-u\right)  ^{k}\mathrm{d}u\int_{\delta_{0}%
}^{1}v^{\alpha-1}\left(  1-v\right)  ^{k}\mathrm{d}v=\frac{\Gamma\left(
\alpha\right)  \Gamma\left(  k+1\right)  }{\Gamma\left(  \alpha+k+1\right)
}\int_{\delta_{0}}^{1}v^{\alpha-1}\left(  1-v\right)  ^{k}\mathrm{d}v,
\]
where $\left\{  u,v\right\}  =\left\{  s_{4},s_{5}\right\}  $ (depending on
whether $\delta_{1}\geq\delta_{2})$. The second integral is an incomplete Beta
integral. As discussed above, the restriction $\alpha\in\mathbb{Z}_{+}$ leads
to a feasible calculation. We use this simple antiderivative formula:%
\begin{gather}
\frac{d}{dv}\sum_{j=0}^{\alpha-1}\left(  -1\right)  ^{j+1}\frac{\left(
1-\alpha\right)  _{j}}{\left(  k+1\right)  _{j+1}}v^{\alpha-1-j}\left(
1-v\right)  ^{k+j+1}=v^{\alpha-1}\left(  1-v\right)  ^{k},\nonumber\\
\int_{\delta_{0}}^{1}v^{\alpha-1}\left(  1-v\right)  ^{k}\mathrm{d}%
v=\sum_{j=0}^{\alpha-1}\left(  -1\right)  ^{j}\frac{\left(  1-\alpha\right)
_{j}}{\left(  k+1\right)  _{j+1}}\delta_{0}^{\alpha-1-j}\left(  1-\delta
_{0}\right)  ^{k+j+1}. \label{antider}%
\end{gather}
From this point on, we work with each term in the sum separately, in the triple
integrals that remain to be evaluated. At the end, the parts will be summed
over $j$ to get the desired probability.

If $\delta_{1}\leq\delta_{2}$ then $\delta_{1}^{\alpha+k}\delta_{2}^{\alpha
+k}\delta_{0}^{\alpha-1-j}\left(  1-\delta_{0}\right)  ^{k+j+1}=\delta
_{1}^{2\alpha+k-1-j}\left(  \delta_{2}-\delta_{1}\right)  ^{k+j+1}$, otherwise
$\delta_{1}^{\alpha+k}\delta_{2}^{\alpha+k}\delta_{0}^{\alpha-1-j}\left(
1-\delta_{0}\right)  ^{k+j+1}=\delta_{2}^{2\alpha+k-1-j}\left(  \delta
_{1}-\delta_{2}\right)  ^{k+j+1}$. As in the previous section, we restrict to
$0\leq s_{1}\leq\frac{1}{2}$ and set $a=\frac{\left(  1-s_{1}\right)  ^{2}}%
{4},b=\frac{s_{1}^{2}}{4}$.

The typical term $\delta_{1}^{\alpha+k}\delta_{2}^{\alpha+k}\delta_{0}%
^{\alpha-1-j}\left(  1-\delta_{0}\right)  ^{k+j+1}$ produces the integral%
\begin{align*}
&  \int_{0}^{\frac{1}{2}}\mathrm{d}s_{1}\int_{0}^{b}\frac{\mathrm{d}\delta
_{2}}{\sqrt{b-\delta_{2}}}\\
&  \times\left\{  \int_{0}^{\delta_{2}}\frac{\delta_{1}^{2\alpha+k-1-j}\left(
\delta_{2}-\delta_{1}\right)  ^{k+j+1}}{\sqrt{a-\delta_{1}}}\mathrm{d}%
\delta_{1}+\delta_{2}^{2\alpha+k-1-j}\int_{\delta_{2}}^{a}\frac{\left(
\delta_{1}-\delta_{2}\right)  ^{k+j+1}}{\sqrt{a-\delta_{1}}}\mathrm{d}%
\delta_{1}\right\}  \\
&  =\int_{0}^{\frac{1}{2}}\mathrm{d}s_{1}\int_{0}^{b}\frac{\delta_{1}%
^{2\alpha+k-1-j}\mathrm{d}\delta_{1}}{\sqrt{a-\delta_{1}}}\int_{\delta_{1}%
}^{b}\frac{\left(  \delta_{2}-\delta_{1}\right)  ^{k+j+1}\mathrm{d}\delta_{2}%
}{\sqrt{b-\delta_{2}}}\\
&  +\int_{0}^{\frac{1}{2}}\mathrm{d}s_{1}\int_{0}^{b}\frac{\delta_{2}%
^{2\alpha+k-1-j}\mathrm{d}\delta_{2}}{\sqrt{b-\delta_{2}}}\int_{\delta_{2}%
}^{a}\frac{\left(  \delta_{1}-\delta_{2}\right)  ^{k+j+1}}{\sqrt{a-\delta_{1}%
}}\mathrm{d}\delta_{1}\\
&  =B\left(  k+j+2,\frac{1}{2}\right)  \int_{0}^{\frac{1}{2}}\mathrm{d}s_{1}\\
&  \times\left\{  \int_{0}^{b}\left(  b-\delta_{1}\right)  ^{k+j+3/2}%
\frac{\delta_{1}^{2\alpha+k-1-j}\mathrm{d}\delta_{1}}{\sqrt{a-\delta_{1}}%
}+\int_{0}^{b}\left(  a-\delta_{2}\right)  ^{k+j+3/2}\frac{\delta_{2}%
^{2\alpha+k-1-j}\mathrm{d}\delta_{2}}{\sqrt{b-\delta_{2}}}\right\}  \\
&  =\frac{\left(  k+j+1\right)  !}{\left(  \frac{1}{2}\right)  _{k+j+2}}%
\int_{0}^{\frac{1}{2}}\mathrm{d}s_{1}\\
&  \times\int_{0}^{b}v^{2\alpha+k-1-j}\left\{  \left(  b-v\right)
^{k+j+1}\sqrt{\frac{b-v}{a-v}}+\left(  a-v\right)  ^{k+j+1}\sqrt{\frac
{a-v}{b-v}}\right\}  \mathrm{d}v\\
&  =\frac{\left(  k+j+1\right)  !}{\left(  \frac{1}{2}\right)  _{k+j+2}%
}I\left(  2\alpha+k-1-j,k+j+1\right)  .
\end{align*}
(Note $B\left(  m+1,\frac{1}{2}\right)  =\frac{m!}{\left(  \frac{1}{2}\right)
_{m+1}}$.) The proof and statement of the $I\left(  m,n\right)  $ formula are
in Prop.\ref{sumMN}. Then%
\begin{align*}
\Pr\left\{  \det\xi^{PT}\leq0\right\}    & =2c_{\alpha,k}B\left(
\alpha,k+1\right)  \\
& \times\sum_{j=0}^{\alpha-1}\left(  -1\right)  ^{j}\frac{\left(
1-\alpha\right)  _{j}\left(  k+j+1\right)  !}{\left(  k+1\right)
_{j+1}\left(  \frac{1}{2}\right)  _{k+j+2}}I\left(  2\alpha
+k-1-j,k+j+1\right)  .
\end{align*}
By formula (\ref{sumMN})%
\[
I\left(  2\alpha+k-1-j,k+j+1\right)  =\frac{\left(  2\alpha+k-1-j\right)
!^{2}\left(  2\alpha+2k+1\right)  !\left(  \frac{5}{2}\right)  _{k+j}%
}{2^{4\alpha+4k+3}\left(  4\alpha+3k-j+1\right)  !\left(  \frac{5}{2}\right)
_{2\alpha+2k}}.
\]
We combine the various terms to arrive at:%
\begin{align*}
\Pr\left\{  \det\xi^{PT}\leq0\right\}   &  =\sum_{j=0}^{\alpha-1}\frac
{2\Gamma\left(  2\alpha+2k+2\right)  ^{2}\Gamma\left(  2\alpha+k-j\right)
^{2}}{\Gamma\left(  \alpha+k+1\right)  ^{3}\Gamma\left(  \alpha-j\right)
\Gamma\left(  4\alpha+3k-j+2\right)  }\\
&  =\frac{\Gamma\left(  2\alpha+2k+3\right)  ^{2}}{2\Gamma\left(
\alpha+k+2\right)  \Gamma\left(  4\alpha+3k+2\right)  }\\
&  \times\sum_{j=0}^{\alpha-1}\frac{\left(  -1\right)  ^{j}}{\left(
\alpha-1\right)  !\left(  \alpha+k+1\right)  }\left(  \alpha+k+1\right)
_{\alpha-1-j}^{2}\left(  -1-3k-4\alpha\right)  _{j};
\end{align*}
the $j$-sum is a polynomial of degree $2\alpha-3-j$ in $k$ (the factor
$\left(  \alpha+k+1\right)  $ obviously cancels when $j\leq\alpha-2$, and when
$j=\alpha-1$ the last factor equals $-3\left(  \alpha+k+1\right)  \left(
-1-3k-4\alpha\right)  _{\alpha-2}$). This phenomenon is qualitatively similar
to our results for the general $4\times4$ situation. In particular for
$\alpha=1$
\[
\Pr\left\{  \det\xi^{PT}\leq0\right\}  =\frac{2\Gamma\left(  2k+4\right)
^{2}}{\Gamma\left(  k+2\right)  \Gamma\left(  3k+6\right)  }=\frac
{2^{4k+7}\left(  \frac{1}{2}\right)  _{k+2}^{2}}{3^{3k+5}\left(  \frac{1}%
{3}\right)  _{k+2}\left(  \frac{2}{3}\right)  _{k+2}},
\]
and for $\alpha=2$%
\[
\Pr\left\{  \det\xi^{PT}\leq0\right\}  =\frac{2}{3}\frac{\left(  k+6\right)
\Gamma\left(  2k+6\right)  ^{2}}{\Gamma\left(  k+3\right)  \Gamma\left(
3k+9\right)  }.
\]

\subsection{Probability of $\det\xi^{PT}\geq\det\xi$}

Next we compute $\Pr\left\{  \det\xi^{PT}\geq\det\xi\right\}  $ for $\alpha=1$
and the measure $\left(  \det\xi\right)  ^{k}\mathrm{d}\nu_{\alpha}$. As above,
it is neater to work with the complement. From equation (\ref{detdet}),
$\det\xi^{PT}-\det\xi=\left(  \delta_{1}-\delta_{2}\right)  \left(
s_{5}\delta_{2}-s_{4}\delta_{1}\right)  $. To make $\det\xi^{PT}-\det\xi\leq0$
either $\delta_{1}\geq\delta_{2}$ and $s_{4}\geq\frac{\delta_{2}}{\delta_{1}%
}s_{5}$ or $\delta_{1}<\delta_{2}$ and $s_{5}>\frac{\delta_{1}}{\delta_{2}%
}s_{4}$. The first two integrations for $\delta_{1}\geq\delta_{2}$ are%
\begin{align*}
& \left(  \delta_{1}\delta_{2}\right)  ^{k+1}\int_{0}^{1}\left(
1-s_{5}\right)  ^{k}\mathrm{d}s_{5}\int_{\frac{\delta_{2}}{\delta_{1}}s_{5}%
}^{1}\left(  1-s_{4}\right)  ^{k}\mathrm{d}s_{4}\\
& =\frac{\left(  \delta_{1}\delta_{2}\right)  ^{k+1}}{k+1}\int_{0}^{1}\left(
1-s_{5}\right)  ^{k}\left(  1-\frac{\delta_{2}}{\delta_{1}}s_{5}\right)
^{k+1}\mathrm{d}s_{5}\\
& =\frac{\delta_{2}^{k+1}}{k+1}\int_{0}^{1}\left(  1-s_{5}\right)  ^{k}\left(
\delta_{1}-\delta_{2}+\delta_{2}\left(  1-s_{5}\right)  \right)
^{k+1}\mathrm{d}s_{5}\\
& =\frac{1}{k+1}\sum_{j=0}^{k+1}\binom{k+1}{j}\left(  \delta_{1}-\delta
_{2}\right)  ^{j}\delta_{2}^{2k+2-j}\int_{0}^{1}\left(  1-s_{5}\right)
^{2k+1-j}\mathrm{d}s_{5}\\
& =\frac{1}{k+1}\sum_{j=0}^{k+1}\binom{k+1}{j}\frac{\left(  \delta_{1}%
-\delta_{2}\right)  ^{j}\delta_{2}^{2k+2-j}}{2k+2-j};
\end{align*}
similarly for $\delta_{1}<\delta_{2}$ the value of the integral is $\frac
{1}{k+1}\sum_{j=0}^{k+1}\binom{k+1}{j}\frac{\left(  \delta_{2}-\delta
_{1}\right)  ^{j}\delta_{1}^{2k+2-j}}{2k+2-j}$. Proceeding as before and with
the same notations carry out the remaining integrations on the term with fixed
$j$, for $0\leq j\leq k+1$. Assume $0\leq s_{1}\leq\frac{1}{2}$, then we
obtain%
\begin{align*}
&  \int_{0}^{\frac{1}{2}}\mathrm{d}s_{1}\int_{0}^{b}\frac{\mathrm{d}\delta
_{2}}{\sqrt{b-\delta_{2}}}\left\{  \int_{0}^{\delta_{2}}\frac{\delta
_{1}^{2k+2-j}\left(  \delta_{2}-\delta_{1}\right)  ^{j}}{\sqrt{a-\delta_{1}}%
}\mathrm{d}\delta_{1}+\delta_{2}^{2k+2-j}\int_{\delta_{2}}^{a}\frac{\left(
\delta_{1}-\delta_{2}\right)  ^{j}}{\sqrt{a-\delta_{1}}}\mathrm{d}\delta
_{1}\right\}  \\
&  =\int_{0}^{\frac{1}{2}}\mathrm{d}s_{1}\int_{0}^{b}\frac{\delta_{1}%
^{2k+2-j}\mathrm{d}\delta_{1}}{\sqrt{a-\delta_{1}}}\int_{\delta_{1}}^{b}%
\frac{\left(  \delta_{2}-\delta_{1}\right)  ^{j}\mathrm{d}\delta_{2}}%
{\sqrt{b-\delta_{2}}}+\int_{0}^{b}\frac{\delta_{2}^{2k+2-j}\mathrm{d}%
\delta_{2}}{\sqrt{b-\delta_{2}}}\int_{\delta_{2}}^{a}\frac{\left(  \delta
_{1}-\delta_{2}\right)  ^{j}}{\sqrt{a-\delta_{1}}}\mathrm{d}\delta_{1}\\
&  =B\left(  j+1,\frac{1}{2}\right)  \int_{0}^{\frac{1}{2}}\mathrm{d}%
s_{1}\left\{  \int_{0}^{b}\left(  b-\delta_{1}\right)  ^{j+1/2}\frac
{\delta_{1}^{2k+2-j}\mathrm{d}\delta_{1}}{\sqrt{a-\delta_{1}}}+\int_{0}%
^{b}\left(  a-\delta_{2}\right)  ^{j+1/2}\frac{\delta_{2}^{2k+2-j}%
\mathrm{d}\delta_{2}}{\sqrt{b-\delta_{2}}}\right\}  \\
&  =\frac{j!}{\left(  \frac{1}{2}\right)  _{j+1}}\int_{0}^{\frac{1}{2}%
}\mathrm{d}s_{1}\int_{0}^{b}v^{2k+2-j}\left\{  \left(  b-v\right)  ^{j}%
\sqrt{\frac{b-v}{a-v}}+\left(  a-v\right)  ^{j}\sqrt{\frac{a-v}{b-v}}\right\}
\mathrm{d}v\\
&  =\frac{j!}{\left(  \frac{1}{2}\right)  _{j+1}}I\left(  2k+2-j,j\right)  .
\end{align*}
Thus%
\begin{align*}
\Pr\left\{  \det\xi^{PT}\leq\det\xi\right\}   &  =2c_{1,k}\sum_{j=0}%
^{k+1}\binom{k+1}{j}\frac{j!}{\left(  k+1\right)  \left(  2k+2-j\right)
\left(  \frac{1}{2}\right)  _{j+1}}I\left(  2k+2-j,j\right)  \\
&  =2^{2k-1}\frac{\left(  \frac{3}{2}\right)  _{k}^{2}\left(  \frac{5}%
{2}\right)  _{k}}{\left(  k+1\right)  !\left(  \frac{5}{2}\right)  _{2k+1}%
}\sum_{j=0}^{k+1}\frac{\left(  -k-1\right)  _{j}\left(  -4k-6\right)  _{j}%
}{\left(  -2k-2\right)  _{j}\left(  -2k-1\right)  _{j}},
\end{align*}
after simplication. The sum is a truncated $_{3}F_{2}$ series.

With more techical details, we can determine $\Pr\left\{  \det\xi^{PT}\leq
\det\xi\right\}  $ for the measure $\left(  \det\xi\right)  ^{k}\mathrm{d}%
\nu_{\alpha}$ for $\alpha,k=1,2,3,\ldots$. Proceeding as for $\alpha=1$ we
start with the integral%
\[
\left(  \delta_{1}\delta_{2}\right)  ^{\alpha+k}\int_{0}^{1}s_{5}^{\alpha
-1}\left(  1-s_{5}\right)  ^{k}\mathrm{d}s_{5}\int_{s_{5}\delta_{0}}^{1}%
s_{4}^{\alpha-1}\left(  1-s_{4}\right)  ^{k}\mathrm{d}s_{4},
\]
where $\delta_{0}=\frac{\delta_{2}}{\delta_{1}}$ and $0\leq\delta_{2}%
\leq\delta_{1}$; and there is a similar expression when $\delta_{1}\leq
\delta_{2}$. Use formula (\ref{antider}) (with $\delta_{0}$ replaced by
$s_{5}\delta_{0}$) to get the value
\begin{align*}
&  \delta_{2}^{k}\int_{0}^{1}s_{5}^{\alpha-1}\left(  1-s_{5}\right)  ^{k}%
\sum_{j=0}^{\alpha-1}\left(  -1\right)  ^{j}\frac{\left(  1-\alpha\right)
_{j}}{\left(  k+1\right)  _{j+1}}\delta_{2}^{\alpha-1-j}s_{5}^{\alpha
-1-j}\left(  \delta_{1}-\delta_{2}+\left(  1-s_{5}\right)  \delta_{2}\right)
^{k+j+1}\mathrm{d}s_{5}\\
&  =\sum_{j=0}^{\alpha-1}\left(  -1\right)  ^{j}\frac{\left(  1-\alpha\right)
_{j}}{\left(  k+1\right)  _{j+1}}\sum_{i=0}^{k+j+1}\binom{k+1+j}{i}B\left(
2\alpha-j-1,2k+j-i+2\right)  \left(  \delta_{1}-\delta_{2}\right)  ^{i}%
\delta_{2}^{2k+2\alpha-i}\\
&  =\sum_{i=0}^{k+\alpha}\left(  \delta_{1}-\delta_{2}\right)  ^{i}\delta
_{2}^{2k+2\alpha-i}\sigma\left(  k,\alpha,i\right)  ,
\end{align*}
where%
\[
\sigma\left(  k,\alpha,i\right)  =\sum_{j=\max\left(  0,i-k-1\right)
}^{\alpha-1}\frac{k!\left(  \alpha-1\right)  !\left(  2\alpha-2-j\right)
!\left(  2k+j-i+1\right)  !}{i!\left(  \alpha-1-j\right)  !\left(
k+1+j-i\right)  !\left(  2k+2a-i\right)  !}.
\]
For $0\leq i\leq k$%
\[
\sigma\left(  k,\alpha,i\right)  =\frac{k!\left(  2k+1-i\right)  !\left(
2\alpha-2\right)  !}{i!\left(  k+1-i\right)  !\left(  2k+2\alpha-1\right)
!}\sum_{j=0}^{\alpha-1}\frac{\left(  1-\alpha\right)  _{j}\left(
2k+2-i\right)  _{j}}{\left(  2-2\alpha\right)  _{j}\left(  k+2-i\right)  _{j}%
},
\]
and for $k+1\leq i\leq k+\alpha$%
\[
\sigma\left(  k,\alpha,i\right)  =\left\{  \alpha\binom{k+\alpha}{\alpha
}\right\}  ^{-2}\binom{k+\alpha}{i}.
\]

Similarly, the integral equals $\sum_{i=0}^{k+\alpha}\left(  \delta_{2}%
-\delta_{1}\right)  ^{i}\delta_{1}^{2k+2\alpha-i}\sigma\left(  k,\alpha
,i\right)  $ when $\delta_{1}\leq\delta_{2}$. With the same steps as above, we
obtain%
\[
\Pr\left\{  \det\xi^{PT}\leq\det\xi\right\}  =2c_{\alpha,k}\sum_{i=0}%
^{k+\alpha}\sigma\left(  k,\alpha,i\right)  \frac{i!}{\left(  \frac{1}%
{2}\right)  _{i+1}}I\left(  2k+2\alpha-i,i\right)  .
\]
Some simplification may be possible.

\section{Integral formulas\label{intform}}

In this section we compute closed expressions for
\[
I\left(  m,n\right)  :=\int_{0}^{\frac{1}{2}}ds\int_{0}^{b}v^{m}\left\{
\left(  a-v\right)  ^{n}\sqrt{\frac{a-v}{b-v}}+\left(  b-v\right)  ^{n}%
\sqrt{\frac{b-v}{a-v}}\right\}  \mathrm{d}v,
\]
for $m,n=0,1,2,3,\ldots$ and for $\left(  m,n\right)  =\left(  2\alpha
,0\right)  $ with $\alpha>0$; where $a=\left(  \frac{1-s}{2}\right)
^{2},b=\left(  \frac{s}{2}\right)  ^{2}$ so that $0\leq b\leq a\leq\frac{1}%
{4}$. The auxiliary formula
\begin{align*}
S\left(  m,n\right)   &  :=\sum_{i=0}^{m}\frac{\left(  -m\right)  _{i}\left(
n+1\right)  _{i}}{i!\left(  m+n+2\right)  _{i}}\sum_{j=0}^{n}\frac{\left(
-n\right)  _{j}\left(  \frac{1}{2}-n\right)  _{j}}{j!\left(  \frac{1}%
{2}\right)  _{j}}\frac{1}{i+j+\frac{1}{2}}\\
&  =2^{2m+2n}\frac{m!\left(  m+n\right)  !\left(  m+n+1\right)  !\left(
\frac{1}{2}\right)  _{n}}{n!\left(  n+2m+1\right)  !\left(  \frac{1}%
{2}\right)  _{m+n+1}}%
\end{align*}
is proved in \cite[Prop. 2]{DG}.

\begin{proposition}
\label{sumMN}For $m,n=0,1,2,3,\ldots$%
\begin{align*}
&  I\left(  m,n\right)  =\int_{0}^{\frac{1}{2}}ds\int_{0}^{b}v^{m}\left\{
\left(  a-v\right)  ^{n}\sqrt{\frac{a-v}{b-v}}+\left(  b-v\right)  ^{n}%
\sqrt{\frac{b-v}{a-v}}\right\}  \mathrm{d}v\\
&  =2^{-4m-4n-5}B\left(  m+1,n+2\right)  S\left(  m,n+1\right)  \\
&  =2^{-2m-2n-3}\frac{\left(  m!\right)  ^{2}\left(  m+n+1\right)  !\left(
\frac{1}{2}\right)  _{n+1}}{\left(  2m+n+2\right)  !\left(  \frac{1}%
{2}\right)  _{m+n+2}}.
\end{align*}

\end{proposition}

\begin{proof}
Set $z=1-2s,u^{2}=\frac{b-v}{a-v},a=\frac{1}{16}\left(  1+z\right)
^{2},b=\frac{1}{16}\left(  1-z\right)  ^{2}$ then%
\begin{align*}
b-v  &  =\frac{u^{2}z}{4\left(  1-u^{2}\right)  },a-v=\frac{z}{4\left(
1-u^{2}\right)  },\\
v  &  =\frac{\left(  1-z\right)  ^{2}-u^{2}\left(  1+z\right)  ^{2}}{16\left(
1-u^{2}\right)  },\mathrm{d}v=-\frac{uz}{2\left(  1-u^{2}\right)  ^{2}%
}\mathrm{d}u,\\
0  &  \leq z\leq1,0\leq u\leq\frac{1-z}{1+z}.
\end{align*}
For changing the order of integration note that $0\leq z\leq1$ and $0\leq
u\leq\frac{1-z}{1+z}$ is equivalent to $0\leq u\leq1$ and $0\leq z\leq
\frac{1-u}{1+u}$. Thus%
\begin{align*}
&  \left\{  \left(  a-v\right)  ^{n}\sqrt{\frac{a-v}{b-v}}+\left(  b-v\right)
^{n}\sqrt{\frac{b-v}{a-v}}\right\}  \mathrm{d}s\mathrm{d}v\\
&  =-\left\{  \left(  \frac{u^{2}z}{4\left(  1-u^{2}\right)  }\right)
^{n}u^{2}+\left(  \frac{z}{4\left(  1-u^{2}\right)  }\right)  ^{n}\right\}
\frac{z\mathrm{d}z\mathrm{d}u}{4\left(  1-u^{2}\right)  ^{2}}\\
&  =-2^{-2n-2}z^{n+1}\frac{u^{2n+2}+1}{\left(  1-u^{2}\right)  ^{n+2}%
}\mathrm{d}z\mathrm{d}u.
\end{align*}
Then introduce the new variable $y$:
\begin{align*}
z  &  =\frac{1-u}{1+u}y,\frac{z^{n+1}}{\left(  1-u^{2}\right)  ^{n+2}%
}\mathrm{d}z=\frac{y^{n+1}}{\left(  1+u\right)  ^{2n+4}}\mathrm{d}y,\\
v  &  =\frac{1}{16}\left(  1-y\right)  \left(  1-\left(  \frac{1-u}%
{1+u}\right)  ^{2}y\right)  ,\\
0  &  \leq u\leq1,0\leq y\leq1,
\end{align*}
and the integral is transformed to%
\begin{align}
&  2^{-4m-2n-2}\int_{0}^{1}\frac{u^{2n+2}+1}{\left(  1+u\right)  ^{2n+4}%
}\mathrm{d}u\int_{0}^{1}y^{n+1}\left(  1-y\right)  ^{m}\left(  1-\left(
\frac{1-u}{1+u}\right)  ^{2}y\right)  ^{m}\mathrm{d}y\label{int1}\\
&  =2^{-4m-2n-2}B\left(  m+1,n+2\right)  \sum_{i=0}^{m}\frac{\left(
-m\right)  _{i}\left(  n+2\right)  _{i}}{i!\left(  m+n+3\right)  _{i}}\int
_{0}^{1}\frac{u^{2n+2}+1}{\left(  1+u\right)  ^{2n+4}}\left(  \frac{1-u}%
{1+u}\right)  ^{2i}\mathrm{d}u,\nonumber
\end{align}
by use of the binomial expansion and%
\[
\int_{0}^{1}y^{n+1}\left(  1-y\right)  ^{m}y^{i}\mathrm{d}y=B\left(
m+1,n+2+i\right)  =B\left(  m+1,n+2\right)  \frac{\left(  n+2\right)  _{i}%
}{\left(  m+n+3\right)  _{i}}.
\]
Let $w=\dfrac{1-u}{1+u},u=\dfrac{1-w}{1+w}$, then $\mathrm{d}u=-\dfrac
{2}{\left(  1+w\right)  ^{2}}\mathrm{d}w$, $\dfrac{1}{1+u}=\frac{1}{2}\left(
1+w\right)  $, and%
\begin{align*}
\int_{0}^{1}\frac{u^{2n+2}+1}{\left(  1+u\right)  ^{2n+4}}\left(  \frac
{1-u}{1+u}\right)  ^{2i}\mathrm{d}u  &  =2^{-2n-3}\int_{0}^{1}\left\{  \left(
1-w\right)  ^{2n+2}+\left(  1+w\right)  ^{2n+2}\right\}  w^{2i}\mathrm{d}w\\
&  =2^{-2n-3}\int_{0}^{1}2\sum_{j=0}^{n+1}\binom{2n+2}{2j}w^{2j+2i}%
\mathrm{d}w\\
&  =2^{-2n-3}\sum_{j=0}^{n+1}\frac{\left(  -n-1\right)  _{j}\left(
-n+\frac{1}{2}\right)  _{j}}{j!\left(  \frac{1}{2}\right)  _{j}\left(
i+j+\frac{1}{2}\right)  }.
\end{align*}
Thus the integral equals%
\begin{align*}
&  2^{-4m-4n-5}B\left(  m+1,n+2\right)  \sum_{i=0}^{m}\frac{\left(  -m\right)
_{i}\left(  n+2\right)  _{i}}{i!\left(  m+n+3\right)  _{i}}\sum_{j=0}%
^{n+1}\frac{\left(  -n-1\right)  _{j}\left(  -n+\frac{1}{2}\right)  _{j}%
}{j!\left(  \frac{1}{2}\right)  _{j}\left(  i+j+\frac{1}{2}\right)  }\\
&  =2^{-4m-4n-5}B\left(  m+1,n+2\right)  S\left(  m,n+1\right)  .
\end{align*}
The proof is completed by simplifying $B\left(  m+1,n+2\right)  S\left(
m,n+1\right)  $.
\end{proof}

\begin{proposition}
\label{sum2a}For $\alpha>0$,%
\[
I\left(  2\alpha,0\right)  =\frac{\pi}{2^{6+8\alpha}}\frac{\Gamma
(2\alpha+1)^{2}}{\Gamma\left(  2\alpha+\frac{3}{2}\right)  \Gamma\left(
2\alpha+\frac{5}{2}\right)  }.
\]

\end{proposition}

\begin{proof}
Follow the previous proof from the beginning to equation (\ref{int1}) which is
replaced by%
\begin{align*}
&  2^{-8\alpha-2}\int_{0}^{1}\frac{u^{2}+1}{\left(  1+u\right)  ^{4}%
}\mathrm{d}u\int_{0}^{1}y\left(  1-y\right)  ^{2\alpha}\left(  1-\left(
\frac{1-u}{1+u}\right)  ^{2}y\right)  ^{2\alpha}\mathrm{d}y\\
&  =2^{-8\alpha-2}\int_{0}^{1}\frac{u^{2}+1}{\left(  1+u\right)  ^{4}%
}\mathrm{d}u\sum_{i=0}^{\infty}\frac{\left(  -2\alpha\right)  _{i}}{i!}\left(
\frac{1-u}{1+u}\right)  ^{2i}B\left(  i+2,2\alpha+1\right) \\
&  =2^{-8\alpha-2}B\left(  2,2\alpha+1\right)  \sum_{i=0}^{\infty}%
\frac{\left(  -2\alpha\right)  _{i}\left(  2\right)  _{i}}{i!\left(
2\alpha+3\right)  _{i}}\int_{0}^{1}\frac{u^{2}+1}{\left(  1+u\right)  ^{4}%
}\left(  \frac{1-u}{1+u}\right)  ^{2i}\mathrm{d}u.
\end{align*}
The binomial series terminates when $2\alpha\in\mathbb{Z}_{+}$; otherwise from
the asymptotic property of the gamma function it follows that (with some fixed
$m\in\mathbb{Z}_{+}$ and $m-2\alpha>0)$
\begin{align*}
\frac{\left(  -2\alpha\right)  _{m+n}}{\left(  m+n\right)  !}  &
=\frac{\left(  -2\alpha\right)  _{m}}{(m+n)!}\left(  m-2\alpha\right)
_{n}=\frac{\left(  -2\alpha\right)  _{m}}{m!\Gamma\left(  m-2\alpha\right)
}\frac{\Gamma\left(  m-2\alpha+n\right)  }{\Gamma\left(  m+1+n\right)  }\\
&  \sim\frac{\left(  -2\alpha\right)  _{m}}{m!\Gamma\left(  m-2\alpha\right)
}n^{-2\alpha-1},\left(  n\rightarrow\infty\right)  ,
\end{align*}
so by the comparison test the series converges for $0\leq y\leq1$ and $0\leq
u\leq1$ provided $\alpha>0$, and can be integrated term-by-term. Also
$B\left(  2,2\alpha+1\right)  =\frac{1}{\left(  2\alpha+1\right)  \left(
2\alpha+2\right)  }$With the same argument as in the previous proof, with
$n=0$ we obtain%
\begin{align*}
\int_{0}^{1}\frac{u^{2}+1}{\left(  1+u\right)  ^{4}}\left(  \frac{1-u}%
{1+u}\right)  ^{2i}\mathrm{d}u  &  =\frac{1}{8}\left(  \frac{1}{i+\frac{1}{2}%
}+\frac{1}{i+\frac{3}{2}}\right)  =\frac{\left(  i+1\right)  }{4\left(
i+\frac{1}{2}\right)  \left(  i+\frac{3}{2}\right)  }\\
&  =\frac{\left(  i+1\right)  !\left(  \frac{1}{2}\right)  _{i}}{4i!\left(
\frac{1}{2}\right)  _{i+2}}=\frac{\left(  2\right)  _{i}\left(  \frac{1}%
{2}\right)  _{i}}{3\left(  1\right)  _{i}\left(  \frac{5}{2}\right)  _{i}}.
\end{align*}
Then%
\[
I\left(  2\alpha,0\right)  =\frac{2^{-8\alpha-2}}{3\left(  2\alpha+1\right)
\left(  2\alpha+2\right)  }\sum_{i=0}^{\infty}\frac{\left(  -2\alpha\right)
_{i}\left(  2\right)  _{i}\left(  2\right)  _{i}\left(  \frac{1}{2}\right)
_{i}}{i!\left(  2\alpha+3\right)  _{i}\left(  1\right)  _{i}\left(  \frac
{5}{2}\right)  _{i}}.
\]
This is a very-well-poised $_{4}F_{3}$ series which is summable using the
Rogers-Dougall formula (see \cite[16.4.9]{DLMF}):%
\begin{align*}
&  _{5}F_{4}\left(
\genfrac{}{}{0pt}{}{a,\frac{a}{2}+1,b,c,d}{\frac{a}{2},a-b+1,a-c+1,a-d+1}%
;1\right) \\
&  =\frac{\Gamma\left(  a-b+1\right)  \Gamma\left(  a-c+1\right)
\Gamma\left(  a-d+1\right)  \Gamma\left(  a-b-c-d+1\right)  }{\Gamma\left(
a+1\right)  \Gamma\left(  a-b-c+1\right)  \Gamma\left(  a-b-d+1\right)
\Gamma\left(  a-c-d+1\right)  },
\end{align*}
the sum converges if one of $b,c,d$ is a negative integer and it terminates or
$b+c+d-a<1$. The formula applies to our sum with $a=2,b=-2\alpha,c=\frac{1}%
{2},d=\frac{3}{2}$ (that is, $d=a-d+1$), and $b+c+d-a=-2\alpha<1$. The result
is
\begin{align*}
I\left(  2\alpha,0\right)   &  =\frac{2^{-8\alpha-2}}{3\left(  2\alpha
+1\right)  \left(  2\alpha+2\right)  }\frac{\Gamma\left(  2\alpha+3\right)
\Gamma\left(  \frac{3}{2}\right)  \Gamma\left(  \frac{5}{2}\right)
\Gamma\left(  2\alpha+1\right)  }{\Gamma\left(  3\right)  \Gamma\left(
2\alpha+\frac{3}{2}\right)  \Gamma\left(  2\alpha+\frac{5}{2}\right)
\Gamma\left(  1\right)  }\\
&  =\frac{2^{-8\alpha-6}\pi}{\left(  2\alpha+1\right)  \left(  2\alpha
+2\right)  }\frac{\Gamma(2\alpha+3)\Gamma(2\alpha+1)}{\Gamma\left(
2\alpha+\frac{3}{2}\right)  \Gamma\left(  2\alpha+\frac{5}{2}\right)  }\\
&  =\frac{\pi}{2^{6+8\alpha}}\frac{\Gamma(2\alpha+1)^{2}}{\Gamma\left(
2\alpha+\frac{3}{2}\right)  \Gamma\left(  2\alpha+\frac{5}{2}\right)  }.
\end{align*}

\end{proof}

Of course the two formulas agree when $2\alpha\in\mathbb{Z}_{\geq0}$.

\begin{acknowledgments}
We would like to thank  Marcelo A. Marchiolli for bringing his joint X-states paper \cite{Xstates2} to our attention, and suggesting that "those results can be useful for your research team in
the recent/future work".
\end{acknowledgments}

\end{document}